\documentclass[12pt]{article}
\usepackage{amsmath,amsthm,amssymb,graphicx}

\def\N{\mathbb N}
\def\N0{\mathbb N_0}
\def\Z{\mathbb Z}

\def\real{\mathbb R}
\def\complex{\mathbb C}

\def\graph{{\Gamma}}
\def\vertexset{{\mathcal V}}
\def\edgeset{{\mathcal E}}
 \def\tree{{\mathcal T}}

\def\Ex{{\mathbb X}}

\def\group{\mathbb G}
\def\genset{\mathbb S}
\def\free{\mathbb F}
\def\Id{\iota}

\title{\bf Quantum Cayley Graphs \\
for Free Groups}
\author{Robert Carlson \\
Department of Mathematics \\ 
University of Colorado at Colorado Springs \\
rcarlson@uccs.edu}

\newtheorem{thm}{Theorem}[section]
\newtheorem{cor}[thm]{Corollary}
\newtheorem{lem}[thm]{Lemma}
\newtheorem{prop}[thm]{Proposition}

\theoremstyle{definition}

\theoremstyle{remark}

\newcommand{\thmref}[1]{Theorem~\ref{#1}}

\newcommand{\lemref}[1]{Lemma~\ref{#1}}
\newcommand{\propref}[1]{Proposition~\ref{#1}}
\newcommand{\corref}[1]{Corollary~\ref{#1}}
\newcommand{\figref}[1]{Figure~\ref{#1}}

 \numberwithin{equation}{section}

\begin{document}

\maketitle

\date{}

\begin{abstract}
Differential operators $\Delta + q$ are considered on metric Cayley graphs of the finitely generated free groups $\free _M$.
The  function $q$ and the graph edge lengths may vary with the $M$ edge types.
Using novel methods, a set of $M$ multipliers $\mu _m(\lambda )$ depending on the spectral parameter is
found.  These multipliers are used to construct the resolvent and characterize the spectrum.   
\end{abstract}

\vskip 25pt

2010 Mathematics Subject Classification 34B45, 58J50

\vskip 25pt

Keywords: quantum graphs, analysis on graphs, spectral geometry.

\newpage

% \tableofcontents

% \newpage

\section{Introduction}

The interplay between a group action and the spectral analysis of a differential operator invariant under the action is a 
popular theme in analysis.  If the group acts on a metric graph, the operators $-D^2+q$ with invariant $q$  
are obvious candidates for a spectral theoretic analysis.  This work treats operators $-D^2 + q$ on a metric Cayley graph 
$\tree _M$ of the nonabelian free group $\free _M$ on $M$ generators.  These Cayley graphs are regular trees, with each vertex having degree $2M$.   
In the present work the $M$ edge types of $\tree _M$ associated to the generators of $\free _M$
may have different lengths, with even functions $q$ varying over the $M$ edge types.  
Remarkably, novel techniques show that there is a system of $M$ multipliers $\mu _m(\lambda )$, resembling those of Hill's equation \cite{MagW} ,  
which can be used to construct the resolvents of the operators.
Echoing the Hill's equation analysis, the location of the spectrum is encoded in the behavior of the multipliers on the real axis.   

There is a large literature treating various aspects of analysis on symmetric infinite graphs. 
Homogeneous trees were considered as discrete graphs in \cite{Brooks}. 
The quantum graph spectral theory of $-D^2 + q$ on homogenous trees was studied in  \cite{Carlson97}, 
assuming that each edge had length $1$, and that $q(x)$ was the same even function on each edge.
These assumptions meant that the graph admitted radial functions, a structure which facilitated a Hill's equation type analysis of the spectral theory.
The spectral theory of radial tree graphs was considered in \cite{Carlson00} and \cite{NS}.
A sampling of work exploiting this structure includes \cite{BrF, Exner, Frank, Hislop, Sob}.
Certain physical models can also lead to graphs of lattices in Euclidean space where the 
group (e.g. $\Z ^2$) is abelian \cite{KP, Montroll}.

The quantum Cayley graph analysis begins in the second section with a review of quantum graphs and the definition of the self-adjoint Hilbert space
operator $\Delta +q$ which acts by sending $f$ in its domain to $-D^2f + qf$. 
The third section reviews basic material on Cayley graphs, and in particular the Cayley graphs $\tree _M$ of the free groups $\free _M$.
For each edge $e$ of the Cayley graph and each $\lambda \in \complex \setminus [0,\infty )$,
a combination of operator theoretic and differential equations arguments identifies a
one dimensional spaces of 'exponential type' functions which are initially defined on 
half of $\tree _M$.  The translational action of generators of $\free _M$ on subtrees of $\tree _M$ induces linear maps on the one-dimensional spaces of exponential functions,
thus producing multipliers $\mu _m(\lambda )$ for $m = 1,\dots ,M$.  
A square integrability condition shows that $|\mu _m(\lambda )| \le 1$ for all $\lambda \in \complex \setminus [0,\infty )$.       

The fourth section starts by linking the multipliers and rather explicit formulas for the resolvent of $\Delta + q$.
Recall that the multipliers for the classical Hill's equation satisfy quadratic polynomial equations with coefficients which are entire 
functions of the spectral parameter $\lambda $.  In this work the multipliers $\mu _m(\lambda )$ satisfy a coupled system
of quadratic equations with coefficients that are entire functions of $\lambda $.  An elimination procedure shows that the equations can be decoupled, 
leading to higher order polynomial equations with entire coefficients for individual multipliers $\mu _m(\lambda )$.
The multipliers $\mu _m(\lambda )$ have extensions from above and below to real  $\sigma $.
The extension is generally holomorphic, but as in the classical Hill's equation the difference $\delta _m(\sigma )$ of the limits from above and below can 
be nonzero.  Except for a discrete set, the spectrum of $\Delta +q$ is characterized by the condition $\delta _m(\sigma ) \not= 0$ for some $m$.

In the final section the system of multiplier equations is explicitly decoupled for  the case $M=2$.
Computer based calculations are used to generate several spectral plots. 
 
\section{Quantum graphs}

Suppose $\graph $ is a locally finite graph with a finite or countably infinite vertex set $\vertexset $ and an directed edge set $\edgeset $. 
 In the usual manner of metric graph construction \cite{BK}, a collection of intervals $\{ [0,l_e], e \in \edgeset \} $ is indexed by the graph edges.
Consistent with the directions of the graph edges $(v,w)$,
the initial endpoint $v$ is associated with $0$, and $w$ is associated with $l_e$.      
Assume that each unordered pair of distinct vertices is joined by at most one edge.  As a result, the map from the directed graph 
to the undirected graph which simply replaces a directed edge $(u,v)$ with an undirected edge $[u,v]$ is one-to-one on the edges.
A topological graph results from the identification of interval endpoints associated to a common vertex.

The Euclidean metric on the intervals is extended to a metric on this topological 
graph by defining the length of a path joining two points to be the sum of its (partial) edge lengths.
The (geodesic) distance between two points is the infimum of the lengths of the paths joining them. 
The resulting metric graph will also be denoted $\graph $.  

To extend the topological graph $\graph $ to a quantum graph, function spaces and differential operators are included.
A function $f:\graph \to \complex $ has restrictions to components $f_e:[0,l_e] \to \complex $.
Let $L^2(\graph )$ denote the complex Hilbert space $\oplus _e L^2[0,l_e]$ with the inner product
\[\langle f, g \rangle = \int_{\graph } f\overline g
= \sum_e \int_{0}^{l_e} f_e(x)\overline{g_e(x)} \ dx .\]
Given a bounded real-valued function $q$ on $\graph $, measurable on each edge, 
differential operators $ -D^2 + q$ are defined to act component-wise on functions $f \in L^2(\graph )$ in their domains.
The functions $q$ are also assumed to be even on each edge, $q_e(l_e-x) = q_e(x)$.  This assumption plays an important role
as the analysis becomes more detailed. 

Self-adjoint operators acting by $-D^2 + q$ can be defined using standard vertex conditions.
The construction of the operator begins with a domain ${\cal D}_{com}$ of compactly supported continuous functions 
$f \in L^2(\graph )$ such that $f_e'$ is absolutely continuous on each edge $e$,  
and $f_e'' \in  L^2[0,l_e]$.  In addition, functions in ${\cal D}_{com}$ are required to be continuous at the graph vertices,
and to satisfy the derivative condition
\begin{equation} \label{derivcond}
\sum_{e \sim v} \partial _{\nu} f_e(v) = 0,
\end{equation}
where $e \sim v$ means the edge $e$ is incident on the vertex $v$,
and $\partial _{\nu } = \partial /\partial x $ in outward pointing local coordinates.  

Since the addition of a constant will make the potential nonegative, but have only a trivial effect on the spectral theory, the assumption
\begin{equation}
q \ge 0
\end{equation}
is made for convenience.
With the domain ${\cal D}_{com}$, the operators $-D^2+ q$ are symmetric and bounded below, with quadratic form
\begin{equation} \label{qform}
\langle (-D^2+q)f,f \rangle = \int_{\graph} |f'|^2 + q|f|^2 .
\end{equation}
These operators always have a self-adjoint
Friedrich's extension, denoted $\Delta + q$, whose spectrum  is a subset of $[0,\infty )$.
When the edge lengths of $\graph $ have a positive lower bound the Friedrich's extension is the unique self adjoint extension \cite[p. 30]{BK}. 

Say that an edge $e = (v_-,v_+) \in \edgeset $ of a connected graph $\graph $ is a bridge if the removal of (the interior of ) $e$ separates the graph into 
two disjoint connected subgraphs.  If $e$ is a bridge, let $\graph _{e}^{\pm}$ denote the closure of the connected component of $\graph \setminus v_{\mp}$ 
which contains the vertex $v_{\pm}$.  Less formally, $\graph _{e}^{\pm}$ includes $e$, the vertices $v_{\pm}$, and the $v_+$ side of $\graph $.
  
For $\lambda \in \complex \setminus [0,\infty )$, the resolvents $R(\lambda ) = (\Delta + q - \lambda I)^{-1}$ of the self adjoint operators $\Delta +q$ provide special solutions 
of $ -y'' + qy = \lambda y$ on $\graph _{e}^{\pm}$.  Let $\Ex _{e}^{\pm}$ denote 
the space of functions $y_{\pm}:\graph_{e}^{\pm} \to \complex $ 
which (i) satisfy 
\begin{equation} \label{eveqn}
 -y'' + qy = \lambda y
 \end{equation}
on each edge $e \in \graph _{\epsilon}^{\pm} $, 
(ii) are continuous and square integrable on $\graph _{\epsilon }^{\pm}$, 
and (iii) which satisfy the derivative conditions \eqref{derivcond} at each vertex of  
$\graph _{\epsilon }^{\pm}$ except possibly $v_{\mp}$.
 
\begin{lem} \label{bridgelem}
Suppose $e$ is a bridge of the connected graph $\graph $, and $\lambda \in \complex \setminus [0,\infty )$.
Then  $\Ex _{e}^{\pm}$ is one dimensional.
\end{lem}

\begin{proof}
Suppose two linearly independent functions $g_1, g_2$ on $\graph _{e}^+ $ satisfy (i) - (iii). 
Then a nontrivial linear combination $h = \alpha _1g_1+\alpha _2g_2$ would satisfy $h(v_-) = 0$.
Consider altering the domain of $\Delta + q$ by replacing the vertex conditions at $v_-$ by the
condition $f(v_-) = 0$ for each edge incident on $v_-$.  The resulting operator is still self-adjoint and nonnegative on $L^2(\graph )$,
and restricts to a self-adjoint operator on $L^2(\graph_{e}^+)$.
The function $h$ is then a square integrable eigenfunction with eigenvalue $\lambda $, which is impossible. 
A similar argument applies to $\graph _{e}^-$.  Thus $\Ex _{e}^{\pm}$ is at most one dimensional.

Suppose $z$ is a nontrivial solution of the equation 
\begin{equation} \label{zeqn}
-z'' + qz = \overline{\lambda} z
\end{equation}
on the interval $[0,l_{e}]$, and
$f$ is a function in the domain of $\Delta + q$ with support in $e$.
As a function on $[0,l_{e}]$ the function $f$ then satisfies $f(0) = f(l_{e}) = 0$
and $f'(0) = f'(l_e) = 0$.  Integration by parts and the vanishing boundary conditions for $f$ lead to 
\[ 0 = \int_0^{l_{e}} f \overline{[-z'' + qz - \overline{\lambda }z]} \ dx 
=  \int_0^{l_{e}} [-f'' + qf - \lambda f ] \overline{z} \ dx .\]

Since $\lambda $ is not in the spectrum of $\Delta + q$, the resolvent $R(\lambda ) = (\Delta + q - \lambda  I)^{-1}$ maps $L^2(\graph)$
onto the domain of $\Delta +q$.   Extend the functions $z$ by zero to the other edges of $\graph $ to obtain an element of $L^2(\graph)$.
Suppose $f = R(\lambda )z \in L^2(\graph)$ had support in $e $; then the above calculation would give
\[\int_0^{l_{e}} |z|^2 \ dx = 0,\]
which is impossible.  Also, since $z$ does vanish outside of $e$, the function $f = R(\lambda )z $ satisfies
$ -f'' + qf = \lambda f $ on each edge other than $e$.

For $i = 1,2$ there are two independent solutions $z_i$ of \eqref{zeqn} on $e$, and after extension of $z_i$ by zero, there are
two independent functions $f_i = R(\lambda )z_i $.  
As noted above, the functions $f_i$ must be nonzero somewhere on $\graph  \setminus e$; 
without loss of generality suppose $f_1$ is not identically zero on $\graph _{e}^+ \setminus e$.    

An argument by contradiction shows that at least one of the functions $f_1,f_2$ must be nonzero $\graph _{e}^- \setminus e$.
Suppose both $f_1$ and $f_2$ are identically zero on $\graph _{e}^- \setminus e$.  
Define functions $g_i \in \Ex _e^+$ which agree with $f_i$ on $\graph _e^+ \setminus e$,
but which satisfies \eqref{eveqn} on $e$, with initial data 
$g_i(v_+) = f_+(v_+)$ and with $\partial _{\nu} g_i(v_+)$ chosen so the derivative conditions \eqref{derivcond} at $v_+$ are satisfied for $g_i$.  
Since $\Ex _e^+$ is at most one-dimensional,
a nontrivial linear combination $g = \alpha _1g_1+\alpha _2g_2 = 0$ on $\graph _e^+$, and so 
$f = \alpha _1f_1+\alpha _2f_2$ is zero on $\graph \setminus e$.
But the existence of a nontrivial function $f = R(\lambda ) z$ with support in $e$ was ruled out above.

The space $\Ex _e^+$ is then the span of $g_1$, and the case of $\Ex _e^-$ is similar.

\end{proof}

The construction of \lemref{bridgelem} also provides the next result.

\begin{lem} \label{hololem}
Suppose $e$ is a bridge of the connected graph $\graph $, and $\lambda \in \complex \setminus [0,\infty )$.
A basis $g(\lambda )$ of $\Ex _{e}^{\pm}$ may be chosen holomorphically in an open disc centered at $\lambda $,
and real valued if $\lambda \in (-\infty , 0)$.
If $x \in \graph _e^{\pm}$ then $g(x,\lambda )$ is holomorphic.
\end{lem}

\begin{proof}
For $\lambda \in \complex \setminus [0,\infty )$ the resolvent $R(\lambda )$ is a holomorphic operator valued function, so the 
functions $f = R(\lambda )z$ are holomorphic with values in the domain of $\Delta + q$.  
For $x \in \graph ^{\pm} \setminus e$ the evaluations $f(x),f'(x)$ are continuous functionals \cite[p. 191-194]{Kato} on the domain of $\Delta + q$,
so the values $g(v_+) $ and $\partial _{\nu} g_i(v_+)$ from \lemref{bridgelem} are holomorphic, 
as are the $L^2(\graph ^{\pm})$ functions $g(\lambda )$ and the values $g(x,\lambda )$ for $x \in e$.
All of these functions can be chosen to be real valued if $\lambda \in (-\infty , 0)$.
\end{proof}

\section{Quantum Cayley graphs}

\subsection{General remarks}

Suppose $\group $ is a finitely generated group with identity $\Id$.  Let $\genset \subset \group $ be a finite generating set for $\group $,
meaning that every element of $\group $ can be expressed as a product of elements of $\genset $ and their inverses.
Following \cite{Meier}, the Cayley graph $\graph _{\group , \genset }$ for the group $\group $ with generating set $\genset $ 
is the directed graph whose vertex set $\vertexset $ is the set of elements of $\group $.  The edge set of $\graph _{\group , \genset }$  
is the set $\edgeset $ of directed edges $(v,vs)$ with $s \in \genset $, {\it initial vertex} $v \in \group $ and {\it terminal vertex} $vs$.
When confusion is unlikely we will simply write $\graph $ for  $\graph _{\group , \genset }$.
Assume that if $s \in \genset $, then $s^{-1} \notin \genset $.
This condition avoids loops $(v,v\Id )$, and insures that at most one directed edge connects any (unordered) pair of vertices.
We will often consider $\graph $ to have undirected edges $[v,vs]$, with the directions given above available when needed.   

$\group $ acts transitively by left multiplication on the vertices of $\graph $; 
that is, for every $v,w \in \vertexset$ there is a $g \in \group$ such that $w = gv$.  
If $e = (v,vs) \in \edgeset $, then $ge = (gv,gvs) \in \edgeset $, so 
$\group $ also acts on $\edgeset $, although this action is not generally transitive.    
Say that two directed edges $e_1,e_2$ are equivalent if there is a $g \in \group $ such that $e_2 = ge_1$.
The equivalence classes will be called edge orbits of the $\group $ action on $\edgeset $.

\begin{prop}
A set $B$ of directed edges is an edge orbit if and only if there is a unique $s \in \genset $ such that $B = \{ (v,vs), v \in \group \} $.
It follows that the number of edge orbits is the cardinality of $\genset $.
\end{prop}

\begin{proof}
If $v,w \in \group $, then $wv^{-1}(v,vs) = (w,ws)$,  so for a fixed $s \in \genset $ all edges of the form $(v,vs)$ are in the same orbit.
If $g(v,vs_1) = (w,ws_2)$, then $gv = w$ and $gvs_1 = ws_2$, so $ws_1 = ws_2$ and $s_1 = s_2$. 
\end{proof}

\begin{prop} \label{connect}
If $\group $ is generated by the finite set $\genset $, then the undirected graph $\graph _{\group , \genset }$ is path connected.
\end{prop}

\begin{proof}
If our requirements on generating sets are momentarily relaxed and $\genset$ is extended to the set $\genset _1 = \{ s,s^{-1}, s \in \genset \}$, 
then the Cayley graph $\graph _{\group , \genset _1}$ will have a directed path from every
element of $\group $ to $\Id $.  An edge of this graph has one of the forms $(v,vs)$ or $(v,vs^{-1})$.  As an undirected edge,
$[v,vs^{-1}] = [vs^{-1},v] = [vs^{-1}, vs^{-1}s]$, so for every directed edge of $\graph _{\group , \genset _1}$ there is
an undirected edge of $\graph _{\group ,\genset }$ with the same vertices.   Consequently, the undirected graph $\graph _{ \group ,\genset }$ is path connected.
\end{proof}

Cayley graphs $\graph $ can be linked with the spectral theory of differential operators. 
To maintain a strong connection with the group $\group $,  the edges of $\graph $ in the same $\group $ orbit  
will have the same length.  The action of $\group $ on the combinatorial edges may then be extended to the
edges of the metric graph by taking $x \in (0, l_{(v,vs)})$ to $gx = x \in (0,l_{(gv,gvs)})$.
This group action also provides a $\group $ action on the functions $f$ on $\graph $.   
The action simply moves the edge index, so that in terms of function components 
$gf_{ge}(gx) = gf_{ge}(x) = f_e(x)$.
Functions are $\group $-invariant if $f_{ge}(x)  = f_e(x)$ for all directed edges $e$ and all $g \in \group$.
A quantum Cayley graph will be a quantum graph whose underlying combinatorial graph is 
the Cayley graph of a finitely generated group, whose edge lengths are constant on edge orbits, and 
whose differential operator $\Delta + q$ commutes with the group action on functions.
Since there is little chance of confusion, the same notation, e.g. $\graph $, will be used for the corresponding quantum, metric, and combinatorial graphs.

\subsection{Free groups and their graphs $\tree _M$}

This work will focus on Cayley graphs with $\group = \free _M$, the free group \cite{Meier} with rank $M$. 
Recall that the elements of $\free _M$ are equivalence classes of finite length words generated by $M$ distinct symbols 
$s_1,\dots ,s_M$ and their inverse symbols $s_1^{-1},\dots ,s_M^{-1}$.  Two words are equivalent if they have a common
reduction achieved by removing adjacent symbol pairs $s_ms_m^{-1}$ or  $s_m^{-1}s_m$.  The group identity is the empty word class,
the group product of words $w_1,w_2$ is the class of the concatenation $w_1w_2$, and inverses are formed by using
inverse symbols in reverse order, e.g. $(s_2s_3^{-1}s_1)^{-1} = s_1^{-1}s_3s_2^{-1}$.

Given a free group $\free _M$ with generating set $\genset = \{ s_1,\dots ,s_M \} $, let $\tree _M$ denote the corresponding Cayley graph.   
These (undirected) graphs (see \figref{F2graph}) have a simple structure \cite[p. 56]{Meier}.

\begin{prop} The undirected graph $\tree _M$ is a tree whose vertices have degree $2M$.
\end{prop}

\begin{proof}
Suppose $\tree _M$ had a cycle with distinct vertices $w_1,\dots ,w_K$ and edges $[w_K,w_1]$ and $[w_k,w_{k+1}]$ for $k = 1,\dots, K-1$. 
In the undirected graph $\tree _M$ edges extend from $w_k$ by some $s_m$ or $s_m^{-1}$,
so that $w_{k+1} = w_ks_m$ or $w_{k+1} = w_ks_m^{-1}$, so each vertex has degree $2M$. 
The sequence of visited vertices $w_1,\dots ,w_K, w_1$ is described by a word of right multiplications by the generators and their inverse symbols equal to $\Id $ in $\free _M$.  
Since this word can be reduced to the empty word, it must have adjacent symbols $s_ms_m^{-1}$ or $s_m^{-1}s_m$.
This means the vertices $w_1,\dots ,w_K$ are not distinct, so no such cycle exists.  Since $\tree _M$ is connected by \propref{connect} and 
has no cycles, $\tree _M$ is a tree.  
\end{proof}

\begin{figure}
\begin{center}
\includegraphics[height=80mm, width=120mm]{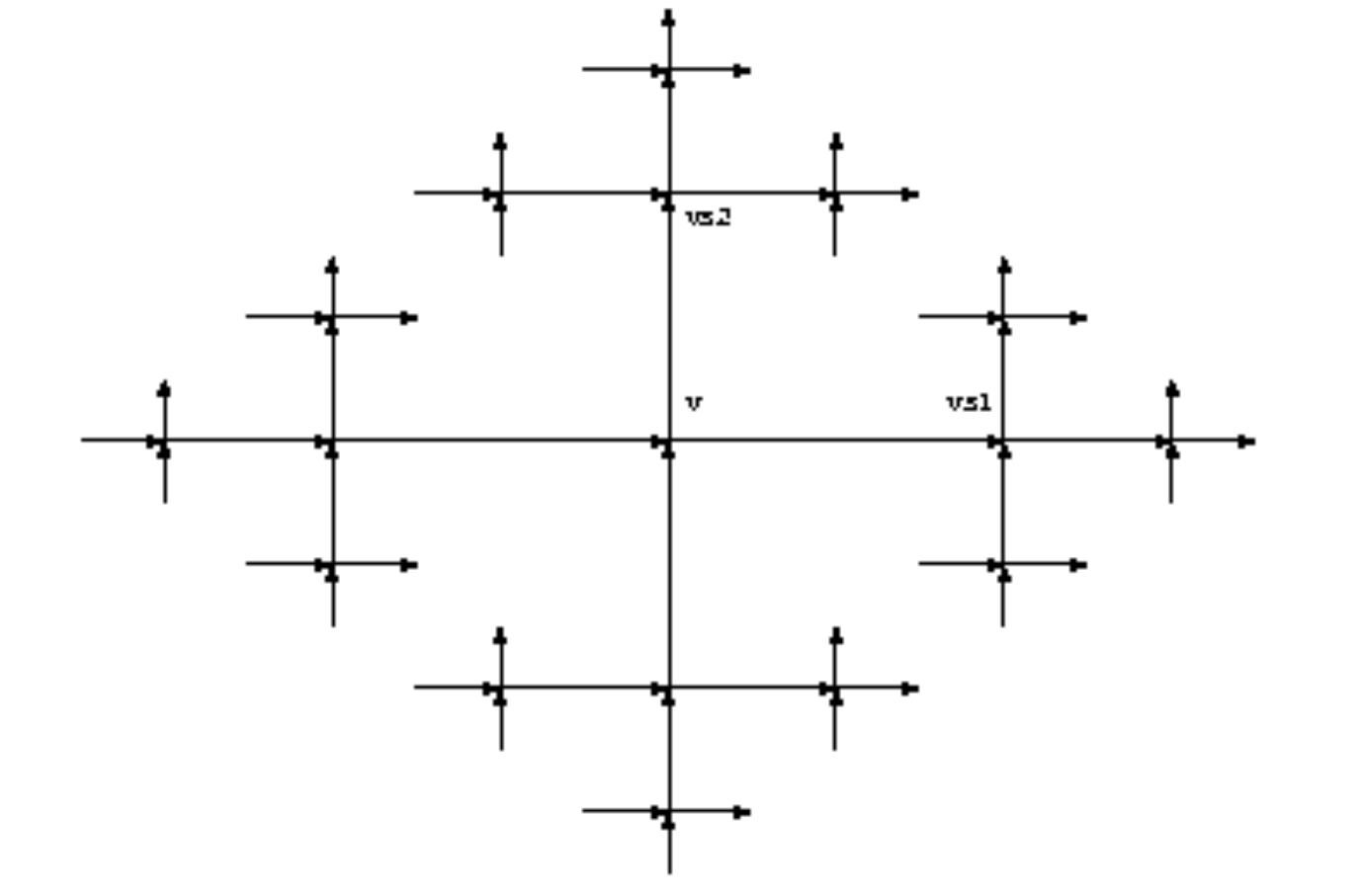}
\caption{A rescaled graph $\tree _2$ } \label{F2graph}
\end{center}
\end{figure}

\subsection{Abelian subgroups and multipliers for $\tree _M$}

Each edge $e = (v,vs) \in \tree _M$ is a bridge.  With $v = v_-$ and $vs = v_+$, the subgraphs $\graph _e^{\pm}$ described above will be
subtrees of $\tree _M$, denoted by $\tree _e^{\pm}$.  
The vector spaces $\Ex _{\epsilon }^{\pm}$ are as in \lemref{bridgelem}.

\begin{lem} \label{nozero}
Suppose $e = (v,vs)$ is an edge of $\tree _M$ and $\lambda \in \complex \setminus [0,\infty )$.
If $y^{\pm}$ is a nontrivial element of $\Ex _{e}^{\pm}$, then $y^{\pm}$ is nowhere vanishing on $\tree _{e}^{\pm}$. 
\end{lem}

\begin{proof}
Suppose $y^{\pm}(x_0) = 0$ for some $x_0 \in \tree _{e}^{\pm}$.  First notice that $y^{\pm}$ must then vanish identically
on the subtree $\tree _0$ consisting of points $x_1$ of $\tree _{e}^{\pm}$ with the property that paths from $x_1$ to $v$ must include
$x_0$.    Otherwise, a nonnegative self-adjoint  operator $\Delta + q$ could be obtained on $L^2(\tree _0)$ by using the
boundary condition $f(x_0) = 0$.
This operator would have a nontrivial square integrable eigenfunction, the restriction of $y^{\pm}$ to $\tree _0$, 
with the eigenvalue $\lambda \in \complex \setminus [0,\infty )$, which is impossible.

Since solutions of $-y'' + qy = \lambda y$ are identically zero on an edge $e$ if $y(x_0) = y'(x_0) = 0$ for some $x_0 \in e$,
we may assume $x_0$ is a vertex.  Since the function $y^{\pm}$ vanishes identically on $\tree _0$, the continuity and derivative conditions
at $x_0$ force $y$ to vanish on all the edges with $x_0$ as a vertex.  The function $y^{\pm}$ must now be identically zero on $\tree _{e}^{\pm}$,
contradicting the assumption that the function was nontrivial. 
\end{proof}

The structure of the elements of $\Ex _{e}^{\pm}$ is strongly constrained by the symmetries of $\tree _{e}^{\pm}$ 
combined with the fact that $\Ex _{e}^{\pm}$ is one dimensional.  A simple observation is the following.

\begin{lem}
Suppose $e_1 = (v_1,v_1s_m)$, $g \in \group $, and $e_2 = (v_2,v_2s_m) = (gv_1,gv_1s_m)$.  For $j=1,2$ let $y_j \in \Ex _{e_j}^{\pm}$
with $y_j(v_j) = 1$.  Then $gy_1 = y_2$.  
\end{lem}

\begin{proof}
The action by $g$ is an isomorphism of $\tree _{e_1}^{\pm}$ and $\tree _{e_2}^{\pm}$.  Since $\Ex _{e_2}^{\pm}$ is one dimensional,
$gy_1$ is a scalar multiple of $y_2$.  These two functions agree at $v_2$, so are equal.
\end{proof}

For each vertex $v$ and integers $k$, left multiplication by the abelian subgroup of elements $vs_m^kv^{-1}$ acts on $\tree _M$.
These maps carry the edge $e(0) =  (v,vs_m)$ to the edges $e(k) = (vs_m^k,vs_m^{k+1})$.  
The key role of these group actions is related to the following geometric observation.

\begin{lem} \label{exhaust}
The trees $\tree _{e(k)}^+$ are nested, with $\tree _{e(k+1)}^+ \subset \tree _{e(k)}^+$. 
In addition, 
\[\tree _M = \bigcup_{k = 0}^{\infty} \tree _{e(-k)}^+ .\]
\end{lem}

\begin{proof}
Other than $vs_m^k$, the vertices of the trees $\tree _{e(k)}^+$ are those elements of $\free _M$ which have a representation $vs_m^{k}s_mg$, 
where $s_mg$ is a reduced word in $\free _M$.  If $w = vs_m^{k+1}s_mg \in \tree _{e(k+1)}^+$ is a vertex with $s_mg$ reduced, then
$s_m^2g$ is reduced and $w = vs_m^{k}s_ms_mg \in \tree_{e(k)}^+$.  Thus the trees $\tree _{e(k)}^+$ are nested.

More generally, for any integer $k$, a word $w \in \free _M$ may be represented as 
$w = vs_m^{k}s_mg$ with $s_mg = s_m^{-k}v^{-1}w $.   First take a reduced representative $u$ for $v^{-1}w$.
Suppose $u$ begins on the left with a string $s_m^j$, followed by an element of $\genset$ different from $s_m^{\pm 1}$. 
Taking $k = j-1$ gives the desired form, and every vertex $w$ is in some $\tree _{e(k)}^+$.

\end{proof}

Suppose $y \in \Ex _{e(0)}^{+}$ satisfies $y(v) = 1$, and $z \in \Ex _{e(1)}^+$ satisfies $z(vs_m) = 1$.
Since $\tree _{e(1)}^+ \subset \tree _{e(0)}^+$, the restriction of $y$ to $\tree _{e(1)}^+ $ is an element of $\Ex_{e(1)}^{+}$.
Because $y$ is nonvanishing, there is a nonzero multiplier $\mu _{m}(\lambda ) \in \complex $ associated to each generator $s_m$
such that  $y = \mu _m(\lambda ) z$ on $\tree _{e(1)}^+ $.  In particular $\mu _m(\lambda ) = y(l_m)/y(0)$.

\begin{lem} \label{holomult}
The multipliers $\mu _m(\lambda )$ are holomorphic for $\lambda \in \complex \setminus [0,\infty )$, with
$\mu _m(\lambda ) =  \overline{\mu _m(\overline{\lambda })}$.
\end{lem}

\begin{proof}
By \lemref{hololem} the formula $\mu _m(\lambda ) = y(l_m)/y(0)$ shows that $\mu _m(\lambda )$ is holomorphic
when $\lambda \in \complex \setminus [0,\infty )$.  If $\lambda \in (-\infty ,0)$ and $y$ is chosen real, then
$\mu (\lambda )$ is real.  The two functions $\mu _m(\lambda )$ and  $ \overline{\mu _m(\overline{\lambda })}$ are holomorphic 
and agree for $\lambda \in (-\infty ,0)$, so agree for all $\lambda \in \complex \setminus [0,\infty )$.
\end{proof}

Because the function $q$ is even on each edge, that is $q(l_m - x) = q(x)$, the same multipliers will arise
when comparing elements of  $\Ex_{e}^{-}$ if the edge directions are reversed by using the generators $s_m^{-1}$ of $\free _M$ instead of  $s_m$. 
These multipliers provide a global extension of functions in $\Ex _{e(0)}^{\pm}$.

\begin{lem}
Suppose $\tree ^+_{e(0)} \subset \tree ^+_{e(j)} \subset \tree ^+_{e(k)} $.  If $y_j \in  \Ex _{e(j)}^+$ with $y_j(v) = 1$, and
$y_k \in  \Ex _{e(k)}^+$ with $y_k(v) = 1$, then $y_j = y_k$ on $\tree _{e(j)}$.   Elements $y^{\pm}$ of  $\Ex _{e(0)}^{\pm}$ 
may be extended via the multipliers to functions defined on all of $\tree _M$. 
\end{lem}

\begin{proof}
The function $y_k$ restricts to an element of  $ \Ex _{e(j)}^+$.  Since nontrivial elements of  $ \Ex _{e(j)}^+$ never vanish, but 
$y_k(v) - y_j(v) = 0$, the difference is the zero element of   $ \Ex _{e(j)}^+$.
Since these extensions of  $\Ex _{e(0)}^+$ are consistent, elements $y^{\pm}$ of  $\Ex _{e(0)}^{\pm}$ extend via the multipliers to 
functions defined on all of the trees $\tree _{e(k)}^+$.   
\end{proof}

\begin{lem} \label{multsize}
For $\lambda \in \complex \setminus [0,\infty )$, the multipliers satisfy $|\mu _m(\lambda )| < 1$.
\end{lem}

\begin{proof}
Recall that $y$ is nowhere vanishing, so 
\[\int_{v}^{vs_m} |y|^2 \not= 0.\] 
A nontrivial element $y$ of $\Ex _e^{\pm}$ is square integrable on $\tree _e^+$, so in particular
\[\sum_{k = 0}^{\infty} \int_{vs_m^k}^{vs_m^{k+1}} |y|^2 = \int_{v}^{vs_m} |y|^2 \sum_{k = 0}^{\infty} |\mu (\lambda )^{2k}| < \infty ,\]
and $|\mu ^{\pm}(\lambda )| < 1$.
\end{proof}

\section{Analysis of the multipliers}

On each edge $[0,l_m]$ the space of solutions to the eigenvalue equation \eqref{eveqn} has a basis $C_m(x,\lambda ),S_m(x,\lambda )$  
satisfying $C_m(0,\lambda ) = 1 = S_m'(0,\lambda )$ and $S_m(0,\lambda ) = 0 = C_m'(0,\lambda )$.
These solutions satisfy the Wronskian identity
\begin{equation} \label{Wronski}
C_m(x,\lambda )S_m'(x,\lambda ) - S_m'(x, \lambda )C_m(x,\lambda ) = 1.
\end{equation}
If $q = 0 $ and $\omega = \sqrt{\lambda }$, these functions are simply $C_m(x,\lambda ) = \cos (\omega x)$, $S_m(x,\lambda ) = \sin (\omega x)/\omega  $.

If $S_m(l_m,\lambda ) = 0$ then $\lambda $ is an eigenvalue for a classical Sturm-Liouville problem, implying $\lambda \in [0,\infty ) $.
For $\lambda \in \complex \setminus [0,\infty )$ there is a unique solution of \eqref{eveqn} with boundary values 
$y_m(0,\lambda ) = \alpha $, $y_m(l_m,\lambda ) = \beta $ given by
\begin{equation} \label{interp}
y_m(x,\lambda ) = \alpha C_m(x,\lambda ) + \frac{\beta - \alpha C_m(l_m,\lambda )}{S_m(l_m,\lambda )} S_m(x,\lambda ).
\end{equation}
 
Because $q_m(x) = q_m(l_m-x)$ for each edge, there is an identity 
\[C_m(l_m - x,\lambda ) = S_m'(l_m,\lambda )C_m(x,\lambda )   - C_m'(l_m,\lambda )S_m(x,\lambda )\]
since both sides of the equation are solutions of \eqref{eveqn} with the same initial data at $x = l_m$.
Setting $x = 0$ leads to the identity
\begin{equation} \label{evenid}
C_m(l_m, \lambda ) = S_m'(l_m,\lambda ) . 
\end{equation}

In addition to the coordinates originally given to the edges of $\tree _M$, it will be helpful to also consider 
local coordinates for $\tree _e^+$ which identify edges with the same intervals $[0,l_k]$, but with the local coordinate
increasing with distance from a given vertex $v$.
Since $q$ is assumed even on each edge, the operators $\Delta +q$ are unchanged
despite the coordinate change.

\subsection{Multipliers and the resolvent}

The next results show that edges in the same orbit have the same multipliers.

\begin{thm} \label{onetype1}
Assume $e = (v,vs_m)$, $\lambda \in \complex \setminus [0,\infty )$, and $y \in  \Ex _{e}^+$ with $y(v) = 1$.   
Suppose the edge  $\epsilon $ in $\tree _e^+$ is in the same edge orbit as $e$,
with the local coordinate for $\epsilon $ increasing with the distance from $v$.
Using the identifications of $e$ and $\epsilon $ with $[0,l_m]$, the restriction $y_1$
of $y$ to $\epsilon $ satisfies 
\[\frac{y_1(l_m)}{y_1(0)}   =  \frac{y(l_m)}{y(0)} = \mu _m(\lambda ).\] 
\end{thm}

\begin{proof}
If $w$ is the vertex of the edge $\epsilon $ closest to $v$ (see \figref{F2Agraph}), then $\epsilon $ has one of the forms $(w,ws_m)$ or 
$(ws_m^{-1},w)$.  In the first case, where $wv^{-1}e = \epsilon $, the tree $\tree _{\epsilon }^+$ is a subtree of $\tree _e^+$, 
and translation by $wv^{-1}$ carries $\Ex _e^+$ to $\Ex _{\epsilon}^+$.  As functions on $[0,l_m]$, $y_1$ is a nonzero multiple of $y$
since $\Ex _e^+$ and $\Ex _{\epsilon}^+$ are one dimensional.  

\begin{figure}
\begin{center}
\includegraphics[height=75mm, width=100mm]{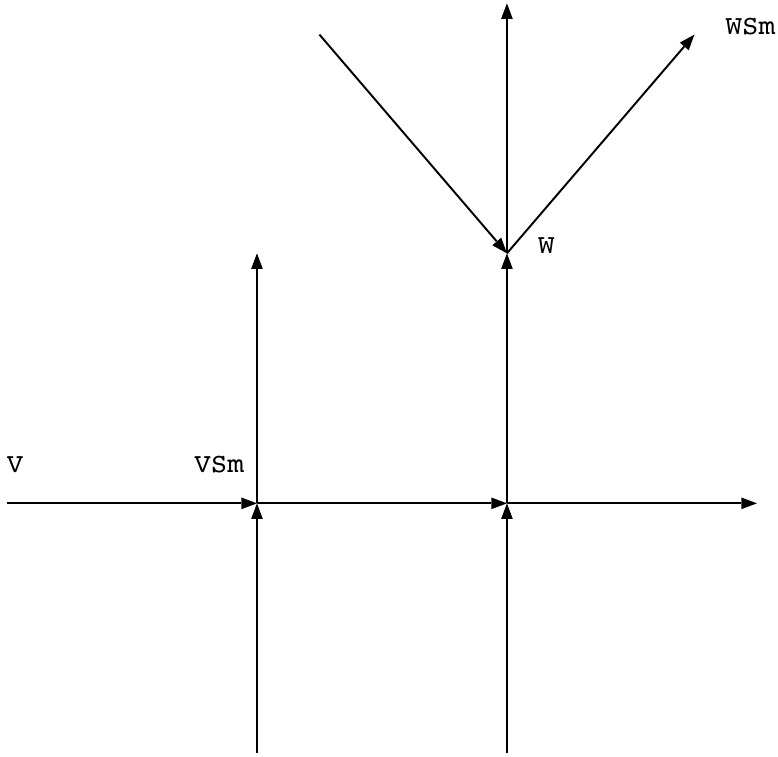}
\caption{Edges in a common orbit} \label{F2Agraph}
\end{center}
\end{figure}

In the second case, when $\epsilon = (ws_m^{-1},w)$, the tree $\tree _{\epsilon}^+$ is generally not a subtree of $\tree _e^+$, 
but $\tree _{(ws_m^{-1},w) }^-$ is.  A different argument will reduce the second case to the first.  
As undirected graphs there are isomorphisms between the trees $\tree _{(w,ws_m)}^+$ and $\tree _{(ws_m^{-1},w)}^-$.
One such is obtained by interchanging the roles of $s_m$ and $s_m^{-1}$ 
There is a corresponding involution of $\Ex _{e}^+$ obtained by 
interchanging function values on the isomorphic trees. 
Since $\Ex _{\epsilon }^+$ is one dimensional, this involution is given by a constant factor.
The nonzero value of $y$ at the vertex $w$ is fixed by the involution, so 
the tree interchange must leave the functions fixed.   
\end{proof}

\begin{cor} \label{onetype2}
Assume $e = (v,vs_m)$, $\lambda \in \complex \setminus [0,\infty )$, and $y \in  \Ex _{e}^+$ with $y(v) = 1$. 
Suppose that for $j = 1,2$, the edges  $\epsilon _j = (w_j,w_js_k)$ in $\tree _e^+$ are in the same edge orbit,
with the local coordinates for $\epsilon _j$ increasing with the distance from $v$.
The restrictions $y_j$ of $y$ to $\epsilon _j$ satisfy 
\[\frac{y_1(l_m)}{y_1(0)}   =  \frac{y_2(l_m)}{y_2(0)} = \mu _k(\lambda ).\] 
\end{cor}

\begin{proof}
If $e_1 = (vs_k^{-1},v)$, then $\tree _e^+ $ is a subtree of $\tree _{e_1}^+$.  
Since $e_1$ and $\epsilon _j$ lie in the same edge orbit, the previous theorem may now be applied.
\end{proof}

As a consequence of \thmref{onetype1} and \corref{onetype2} the functions $y \in \Ex _e^+$ have the following description.

\begin{thm} \label{Yval}
Assume $\lambda \in \complex \setminus [0,\infty )$, $e = (v,vs_m)$ and $y_+ \in \Ex _e^+$ with $y_+(v) = 1$.
Suppose $w$ is a vertex in $\tree _e^+$, and the path from $v$ to $w$ is given by the
reduced word $s_ms_{k(1)}^{\pm 1} \dots s_{k(n)}^{\pm 1}$.  Then 
\begin{equation} \label{yvform}
y_+(w) = \mu _m\mu _{k(1)} \dots \mu _{k(n)} .
\end{equation}
Using \eqref{interp}, the vertex values of $y_+$ can be interpolated to the edges.
\end{thm} 

Because the functions in $\Ex _{e}^+$ are continuous, the multipliers $\mu _k(\lambda )$ are simply the value at
$l_k$ of the solution $y_k$ in $\Ex _{e}^+$ with initial value $1$ at $x = 0$ on edges of type $k$.  That is, 
\begin{equation} \label{multder}
\mu _k(\lambda ) = C_k(l_k,\lambda ) + y_k'(0)S_k(l_k,\lambda ).
\end{equation}

\thmref{Yval} may also be used to describe the functions $y_- \in \Ex _e^-$.  The functions $y_+,y_-$ can be used to construct
the resolvent $R(\lambda ) = [\Delta + q - \lambda ]^{-1}$ on $\complex \setminus [0,\infty) $.  
If the (nonvanishing) functions $y_-$ and $y_+$ were linearly dependent on $e$,
then there would is a nonzero constant $c$ such that $y_-(x) = cy_+(x)$ for $x \in e$, and the function
\[\Bigl \{ \begin{matrix} y_-(x),& x \in \tree _e^-, \cr 
cy_+(x),& x \in \tree _e^+, 
\end{matrix} \Bigr \}\]
would be a square integrable eigenfunction for $\Delta + q$.
Consequently, the functions $y_-$ and $y_+$ must be linearly independent on $e$ if $\lambda \in \complex \setminus [0,\infty )$.

In particular for each $\lambda \in \complex \setminus [0,\infty )$ the Wronskian $W_k(\lambda ) = y_+y_-' - y_+'y_-$  for edges of each type $k$ 
is nonzero, and independent of $x$.  By using \eqref{multder} the Wronskian $W_k(\lambda )$ can be expressed in terms of the multipliers.
Consider evaluation of $W_k(\lambda )$ at $x= l_k$.  Compared to $y_+$, which satisfies \eqref{multder}, with $y_+(0,\lambda ) = 1$,
the function $y_-$ would have the edge direction reversed.  This function has $y_-(l_k) = 1$, and because of the reversed edge direction,
\[y_-(x,\lambda ) = C_k(l_k - x,\lambda )  - y_-'(l_k)S_k(l_k-x,\lambda ),\] so that
\[\mu _k(\lambda ) = y_-(0,\lambda ) = C_k(l_k,\lambda ) - y_-'(l_k)S_k(l_k,\lambda ).\]
Evaluation at $x= l_k$ gives
\[ W_k(\lambda ) = (y_+y_-' - y_+'y_-)(l_k) \]
\[ = \mu _k (\lambda ) \frac{C_k(l_k,\lambda ) - \mu _k(\lambda )}{S_k(l_k,\lambda )}
- [C_k'(l_k,\lambda )+ \frac{\mu _k(\lambda ) - C_k(l_k,\lambda )}{S_k(l_k,\lambda )} S_k'(l_k,\lambda )] \]
and the identities \eqref{Wronski} and \eqref{evenid} give the simplification 
\begin{equation} \label{Weval}
W_k(\lambda ) =  \frac{1 - \mu _k^2(\lambda )}{S_k(l_k,\lambda )}. 
 \end{equation}

For $\lambda  \in \complex \setminus [0,\infty )$ and $t \in [0, l_e]$, define the kernel
\begin{equation} \label{rkernel}
R_e(x,t,\lambda  ) = 
\begin{matrix}
y_-(x,\lambda  )y_+(t,\lambda  )/W_k(\lambda ), \quad 0 \le x \le t \le l_e, \cr 
y_-(t, \lambda  )y_+(x, \lambda  )/W_k(\lambda ), \quad 0 \le t \le x \le l_e.  
\end{matrix} 
\end{equation}
The following observations show that the values of $x$ can be extended to the whole of $\tree _M$.

If $f_e$ is supported in the interior of $e$ the function 
$$h_e(x) = \int_0^{l_e} R_e(x,t,\lambda  )f_e(t) \, dt $$
satisfies $[\Delta +q-\lambda ]h_e = f_e$, and in neighborhoods of $0$ and $l_e$ the function $h_e$  
satisfies \eqref{eveqn}.
The kernel $R_e(x,t,\lambda )$ and the function $h_e$ can then be extended to $\tree _M$ using the values of 
$y_-$ and $y_+$ on $\tree _e^{\pm}$.  The extended function $h_e$ 
is square integrable on $\tree _M$ and satisfies the 
vertex conditions, so $h_e$ agrees with the image of the resolvent acting on $f_e$, that is
$h_e = R(\lambda  )f_e$.  Since the linear span of functions $f_e$ is dense in  $L^2(\tree _M)$,
and the resolvent is a bounded operator for  $\lambda  \in \complex \setminus [0,\infty )$, the discussion above implies the next result.

\begin{thm} \label{theres} 
For $\lambda  \in \complex \setminus [0,\infty )$,
\[R(\lambda  )f = \sum_e \int_0^{l_e} R_e(x,t,\lambda  )f_e(t) \, dt , \quad f \in L^2(\tree _M), \] 
the sum converging in $L^2({\tree _M})$. 
\end{thm}
 
 \subsection{Equations for the multipliers}
 
 \begin{thm} \label{muform}
 For $\lambda \in \complex \setminus [0,\infty )$ and $m = 1,\dots ,M$, the multipliers $\mu _m(\lambda )$ satisfy the system of equations
 \begin{equation} \label{recurse3}
\frac{\mu _m^2(\lambda ) - 1 }{S_m(l_m,\lambda )\mu _m(\lambda )}  - 2\sum_{k = 1}^M \frac{\mu _k(\lambda ) - C_k(l_k,\lambda )}{S_k(l_k,\lambda )} =0 . 
 \end{equation}
 \end{thm}

\begin{proof}

Begin with an edge $e = (v,vs_m)$.   In addition to $e$, the vertex $vs_m$ has $2M-1$ other incident edges.
One is the type $m$ edge $(vs_m,vs_m^2)$, while the others have one of the type $k$ forms
$(vs_m,vs_ms_k)$ or $(vs_ms_k^{-1},vs_m)$, where $k \in \{ 1,\dots ,M \} \setminus \{ m \} $.  
Let $ y \in \Ex _{e}^+$ satisfy $y(vs_m) = 1$, and let $y_k$, respectively $y_m$ denote the restriction of $y$
to the subtrees with root vertex $vs_m$ and initial edges of type $k$, respectively $m$.
As a consequence of \corref{onetype2}, for a fixed value of $k$ the two functions $y_k$ agree as functions of the distance from $vs_m$
on the two edges of type $k$ incident on $vs_m$.

Let $z_m$ denote the restriction of $y$ to the edge $e$.  The 
continuity and derivative vertex conditions at $vs_m$ relate $z_m$ to the restrictions $y_k,y_m$.
Using local edge coordinates which identify $vs_m$ with $l_m$ for the edge $e$, and which identify $vs_m$ with $0$ for the
other incident edges, the initial data for $z_m$ at $vs_m$ is
\[z_m(l_m,\lambda ) = 1, \quad z_m'(l_m,\lambda ) = y_m'(0,\lambda ) + 2\sum_{k \not= m} y_k'(0,\lambda ). \]
The function $z_m(x,\lambda )$ may be written as
\[z_m(x,\lambda ) = C_m(l_m - x, \lambda ) - [y_m'(0,\lambda ) + 2\sum_{k \not= m} y_k'(0,\lambda )]S_m(l_m-x,\lambda ).\]
Evaluation at $x = 0$ gives
\[z_m(0,\lambda ) = C_m(l_m,\lambda ) - [y_m'(0,\lambda ) + 2\sum_{k \not= m} y_k'(0,\lambda )]S_m(l_m,\lambda ).\]
\[z_m'(0,\lambda ) = -C_m'(l_m,\lambda ) + [y_m'(0,\lambda ) + 2\sum_{k \not= m} y_k'(0,\lambda )]S_m'(l_m,\lambda ).\]

As noted above, on $[0,l_m]$ the function $z_m(x,\lambda)$ is a scalar multiple of $y_m(x,\lambda )$,
so that $z_m(x,\lambda )/z_m(0,\lambda ) = y_m(x,\lambda )$.  In particular,
\[y_m'(0,\lambda ) = \frac{z_m'(0,\lambda )}{z_m(0,\lambda )} \]
\[= \frac{-C_m'(l_m,\lambda ) + [y_m'(0,\lambda ) + 2\sum_{k \not= m} y_k'(0,\lambda )]S_m'(l_m,\lambda )}
{C_m(l_m,\lambda ) - [y_m'(0,\lambda ) + 2\sum_{k \not= m} y_k'(0,\lambda )]S_m(l_m,\lambda )}.\]
This can be rewritten as 
\begin{equation} \label{recurse}
 C_m(l_m,\lambda )y_m'(0,\lambda ) + C_m'(l_m,\lambda ) 
 \end{equation}
\[ = [y_m'(0,\lambda ) + 2\sum_{k \not= m} y_k'(0,\lambda )][S_m(l_m,\lambda )y_m'(0,\lambda ) + S_m'(l_m,\lambda )] . \]
 
Using \eqref{multder} to substitute for $y_k'(0)$ in \eqref{recurse} gives
\begin{equation} 
 C_m(l_m,\lambda )\frac{\mu _m(\lambda ) - C_m(l_m,\lambda )}{S_m(l_m,\lambda )} + C_m'(l_m,\lambda ) 
 \end{equation}
 \[ = [\frac{\mu _m(\lambda ) - C_m(l_m,\lambda )}{S_m(l_m,\lambda )} + 2\sum_{k \not= m} \frac{\mu _k(\lambda ) - C_k(l_k,\lambda )}{S_k(l_k,\lambda )} ] \]
\[ \times [\mu _m(\lambda ) - C_m(l_m,\lambda ) + S_m'(l_m,\lambda )] .\]
With the help of the identities \eqref{Wronski} and \eqref{evenid}, these equations can be rewritten as \eqref{recurse3}.

\end{proof}

The solutions $\mu _1(\lambda ),\dots ,\mu _M(\lambda )$ of \eqref{recurse3} coming from the resolvent of $\Delta +q$
can be recognized by a square summability condition. 

\begin{thm} \label{intmult}
Suppose $\lambda \in \complex \setminus [0,\infty )$, $e = (v,vs_m)$, and $w$ denotes a vertex in $\tree _e^+$.
Assume that $\mu _1(\lambda ),\dots ,\mu _m(\lambda )$ satisfy \eqref{recurse3}.
Define a function $y_+ : \tree _e^+ \to \complex $ by taking $y_+(v) = 1$, defining $y_+(w)$ by \eqref{yvform}, and 
interpolating the vertex values of $y_+$ to the edges using \eqref{interp}.

The function $y_+$ is square integrable on $\tree _e^+$ if and only if 
\begin{equation} \label{summcond0}
\sum_{w \in \tree _e^+} |y_+(w)|^2 < \infty .
\end{equation}

If $y_-:\tree _e^- \to \complex $ is defined similarly, and \eqref{summcond0} is satisfied, then the formula
\eqref{rkernel} gives the resolvent of $\Delta +q$ as in \thmref{theres}.
\end{thm} 

\begin{proof}
The solutions of \eqref{recurse3} satisfy $\mu _m(\lambda ) \not= 0$.
For $j = 1,2$, two edges $(w_j,w_js_m) \in \tree _e^+$ have vertex values satisfying $y_+(w_js_m)/y_+(w_j) = \mu _m(\lambda )$,
so the interpolated edge values $y_j$ given by \eqref{interp} satisfy
\[y_2(0) y_1(x) = y_1(0)y_2(x), \quad 0 \le x \le l_m.\]
As a result,
\[\int_0^{l_m} |y_2(x)|^2 \ dx = \frac{|y_2(0)|^2}{|y_1(0)|^2}\int_0^{l_m} |y_1(x)|^2 \ dx ,\]
so by comparing edges of type $m$ with one of the edges incident on $vs_m$ we see that 
$y_+$ is square integrable on $\tree _e^+$ if and only if \eqref{summcond0} holds.

Running the argument of \thmref{muform} in reverse shows that the continuity and derivative conditions \eqref{derivcond} hold 
at the vertices of $\tree _e^+$ except possibly at $v$.  If \eqref{summcond0} holds, then $y_+ \in \Ex _e^+$ and the claims
about the resolvent formula follow.
\end{proof}

The equations \eqref{recurse3} have implications for the decay of $\mu _m(\lambda )$.

\begin{thm}
For $ 0 < \sigma < \pi$, let $\Omega $ denote the set of $\lambda $ with $| \arg (\lambda ) - \pi | \le \sigma $ 
and $|\lambda | \ge 1$.  For $\lambda \in \Omega $,
\begin{equation} \label{asympt2}
\lim_{|\lambda | \to \infty }  |\mu _m(\lambda )| e^{|\Im (\sqrt{\lambda })|l_m}  =  \frac{1}{2M}. 
\end{equation}

\end{thm}

\begin{proof}
Rewrite \eqref{recurse3} as 
\begin{equation} \label{recurse4}
\frac{\mu _m (\lambda )}{\sqrt{\lambda } S_m(\lambda, l_m)} - 
\frac{1}{\sqrt{\lambda }S_m(\lambda , l_m)\mu _m(\lambda )}  = 2\sum_{k = 1}^M \frac{\mu _k(\lambda ) - C_k(\lambda ,l_k)}{\sqrt{\lambda }S_k(\lambda ,l_k)}  . 
\end{equation}

Take $\Im (\sqrt{\lambda }) > 0$ and recall that $|\mu _k(\lambda )| \le 1$.
The functions $C_k(\lambda ,l_k), S_k(\lambda ,l_k)$ satisfy the estimates \cite[p. 13]{PT}
\[ |C_k(\lambda , l_k) - \cos (\sqrt{\lambda } l_k)| \le \frac{C}{|\sqrt{\lambda }| }\exp(|\Im (\sqrt{\lambda })|l_k), \]
\[ |S_k(\lambda , l_k) - \frac{\sin (\sqrt{\lambda } l_k)}{\sqrt{\lambda }} | \le \frac{C}{|\lambda |}\exp(|\Im (\sqrt{\lambda })l_k),\]
while
\[\cos (\sqrt{\lambda }x) = \frac{1}{2} e^{-i \sqrt{\lambda } x} (1 + e^{2i  \sqrt{\lambda}x} ), \quad
\sin (\sqrt{\lambda }x) = \frac{i}{2}e^{-i \sqrt{\lambda } x}(1 - e^{2i  \sqrt{\lambda}x} ).\]
For $\lambda \in \Omega $, taking $|\lambda | \to \infty $ in \eqref{recurse4} gives 
\[\lim_{|\lambda | \to \infty }  \frac{1}{\sqrt{\lambda } S_m(l_m)\mu _m(\lambda )}  =  -2iM,\]
which implies \eqref{asympt2}.
\end{proof}

For each $\lambda \in \complex $ the equations \eqref{recurse3} are a system of polynomial equations 
$P_m(\xi _1,\dots , \xi _M,\lambda ) = 0$ for $m = 1,\dots ,M$ which are satisfied by the multipliers $\mu _1,\dots ,\mu _M$.
The independence of the equations and local structure of the solutions may be determined by
computing the gradients $\nabla P_m$ with respect to $\xi _1,\dots ,\xi _M$, with $\lambda $
treated as a parameter.  Recall from \lemref{multsize} that the multipliers satisfy $|\mu _m(\lambda )| < 1$ if 
$\lambda \in \complex \setminus [0,\infty )$.

\begin{thm} \label{grads}
Suppose $\lambda \in \complex$, $S_m(\lambda, l_m) \not= 0$, and $\xi _m(\lambda )^2 + 1 \not = 0$
for $m=1,\dots, M$.  Then the complex gradients $\nabla P_m$ are linearly independent if 
\begin{equation} \label{singset}
\sum_{m=1}^M \frac{\xi _m^2}{\xi _m^2 + 1} \not= 1/2.
\end{equation}
\end{thm}

\begin{proof}
The relevant partial derivatives are 
\[\frac{\partial P_m}{\partial \xi _m} =
\frac{1}{S_m(l_m)}[\frac{1}{\xi _m^2(\lambda ) } - 1] \]
and for $j \not= m$
\[\frac{\partial P_m}{\partial \xi _j} = -2 \frac{1}{S_j(l_j)}.\]
That is, there is an $m$-independent vector function $W$ such that 
\[\nabla P_m = V_m + W, \quad W = -2 \begin{pmatrix}  1/S_1(l_1) \cr \vdots \cr  1/S_M(l_M) \end{pmatrix}, \]
with $V_m$ having $m$-th component equal to 
\[\partial P_m/\partial \xi _m + 2/S_m(l_m) = \frac{1}{S_m(l_m)}[\frac{1}{\xi _m^2(\lambda ) } + 1] \]
and all other components zero.

If the vectors $V_k + W$ are linearly dependent, then there are constants $\alpha _k$
not all zero such that $\sum_k \alpha _k (V_k + W) = 0$.  
 If $S_m(l_m) \not= 0$ the component equations can be written as 
\[\alpha _m[\frac{1}{\xi _m(\lambda )^2} + 1] = 2\sum_{k=1}^M \alpha _k .\]
This linear system is 
\[ {\rm diag}[ \frac{\xi _1^2 + 1}{\xi _1^2}, \dots ,  \frac{\xi _M^2 + 1}{\xi _M^2}] 
\begin{pmatrix} \alpha _1 \cr \vdots \cr \alpha _M \end{pmatrix} =
2 \begin{pmatrix} 1 & \dots & 1 \cr \vdots & \dots & \vdots \cr 1 & \dots & 1\end{pmatrix}
 \begin{pmatrix} \alpha _1 \cr \vdots \cr \alpha _M \end{pmatrix} .\]

If none of the $\xi _m^2$ have the value $-1$, this system says $[\alpha _1,\dots ,\alpha _M]$ is an
eigenvector with eigenvalue $1/2$ for the matrix
\[ \begin{pmatrix} \frac{\xi _1^2}{\xi _1^2 + 1} & \dots &  \frac{\xi _1^2}{\xi _1^2 + 1}
 \cr \vdots & \dots & \vdots \cr  \frac{\xi _M^2}{\xi _M^2 + 1} & \dots &  \frac{\xi _M^2}{\xi _M^2 + 1} \end{pmatrix} .\]
Vectors with $\sum \alpha _k = 0$ are in the null space of this matrix, and the remaining eigenvalue is
the trace, with eigenvector $[\frac{\xi _1^2}{\xi _1^2 + 1}, \dots , \frac{\xi _M^2}{\xi _M^2 + 1}]$.
Thus the condition for dependent gradients is
\[\sum_{m=1}^M \frac{\xi _m^2}{\xi _m^2 + 1} = 1/2.\]

\end{proof}

By applying the inverse and implicit function theorems for holomorphic functions \cite[p. 18-19]{GH} we obtain the following corollary.

\begin{cor} \label{localman}
Suppose $\lambda \in \complex$, $S_m(\lambda, l_m) \not= 0$, and 
\[ \xi _m(\lambda )^2 + 1 \not = 0, \quad \sum_{m=1}^M \frac{\xi _m^2}{\xi _m^2 + 1} \not= 1/2, \quad m=1,\dots, M.\]
Then the solutions of the system \eqref{recurse3} are locally given in $\complex ^M \times \complex $ 
by a holomorphic $\complex ^M$- valued function of $\lambda $.  
\end{cor}

\begin{thm}
There is a discrete set $Z_0 \subset \real  $ and a positive integer $N$ such that for all $\lambda \in \complex \setminus Z_0$ the equations \eqref{recurse3}
satisfied by the multipliers $\mu _m(\lambda )$ have at most $N$ solutions $\xi _1(\lambda ),\dots ,\xi _M(\lambda )$.   
For $\lambda \in \complex \setminus Z_0$, the functions $\mu _m(\lambda )$ are solutions of 
polynomial equations $p_m(\xi _m) = 0$ in the one variable $\xi _m$ of positive degree, with coefficients which are entire functions of $\lambda $.
\end{thm}

\begin{proof}
The polynomial equation in the single variable $\xi _1(\lambda )$ will be considered;  
the equations satisfied by the other functions $\mu _m(\lambda )$ may be treated in the same manner.

Notice that the system \eqref{recurse3}  has the form
\[F_1(\xi _1,\lambda ) = g(\xi _1,\dots ,\xi _M,\lambda) , \dots  ,F_M(\xi _1,\lambda ) = g(\xi _1,\dots ,\xi _M,\lambda),\]
with 
 \begin{equation} \label{Fform1}
F_m(\xi _m,\lambda ) = \frac{\xi _m^2(\lambda ) - 1 }{S_m(l_m)\xi _m(\lambda )},
 \end{equation}
\[g(\xi _1,\dots ,\xi _M,\lambda) =  2\sum_{k = 1}^M \frac{\xi _k(\lambda ) - C_k(l_k)}{S_k(l_k)} .\] 
Subtraction of successive equations eliminates the right hand sides from $M-1$ equations, giving a system of $M$ equations,
the equations indexed by the value of $j = 1,\dots ,M$,
\[\begin{matrix}j = 1 & F_1(\xi _1,\lambda ) = g(\xi _1,\dots ,\xi _M,\lambda) \cr
 j = 2 & F_1(\xi _1,\lambda ) - F_2(\xi _2,\lambda ) = 0, \cr
\vdots \cr
j = M & F_{M-1}(\xi _{M-1},\lambda ) - F_M(\xi _M,\lambda ) = 0. \end{matrix} \]

For $m > 1$, the $m$-th equation can be written as
\begin{equation} \label{sqform}
\xi _m(\lambda )^2 =  \Bigl ( \frac{\xi _{m-1}^2(\lambda ) - 1 }{S_{m-1}\xi _{m-1}(\lambda )} \Bigr )S_m\xi _m + 1,
\end{equation}
or by using the quadratic formula,
\begin{equation} \label{linform}
2\xi _m(\lambda ) -  \Bigl ( \frac{\xi _{m-1}^2(\lambda ) - 1 }{S_{m-1}\xi _{m-1}(\lambda )}  \Bigr )S_m
= \Bigl [ \Bigl ( \frac{\xi _{m-1}^2(\lambda ) - 1 }{S_{m-1}\xi _{m-1}(\lambda )}  \Bigr )^2S_m^2  + 4 \Bigr ]^{1/2} .
\end{equation}

The variables $\xi _M, \dots , \xi _2$ can be successively eliminated from the first equation.  Starting with $k = M$ and continuing up the list of indices, 
the first equation can be written as a polynomial equation for $\xi _k$ with coefficients which are polynomials in $\xi _1,\dots ,\xi _{k-1}$
and the entire functions $S_m(l_k,\lambda )$.  
Repeated use of the substitution \eqref{sqform}, followed by clearing of the denominators, reduces the first equation to degree one in $\xi _k$.
These substitutions result in an equivalent system of equations as long as $\lambda \notin Z_0$, where
\begin{equation} \label{Z0}
Z_0 = \bigcup_m \{ S_m(l_m, \lambda ) = 0 \} .
\end{equation}
Now solve for $2\xi _k$, subtract $ [\xi _{m-1}^2(\lambda ) - 1 ] S_m/[S_{m-1}\xi _{m-1}(\lambda )] $,
use the substitution \eqref{linform}, and square both sides.  Since squaring is a two-to-one map, it will not change the dimension
of the set of solutions.  After applying these substitutions, 
the variable $\xi _k$ has been eliminated, and after clearing the denominators, the modified $j=1$
equation is a polynomial in $\xi _1(\lambda ), \dots , \xi _{k-1}(\lambda )$ with entire coefficients. 
The substitution process provides a common bound $N$ for the degrees of the polynomials $p_m$. 

Suppose the final version of the first equation does not have positive degree for $\xi _1$.
Define $Q_m = F_m(\xi _m,\lambda ) - F_{m+1}(\xi _{m+1},\lambda ) $ for $m = 1,\dots ,M-1$.
The system of equations for $\xi _1,\dots ,\xi _M$ is then the system
\[Q_1 = 0, \dots , Q_{M-1} = 0,\]
where
\[\frac{\partial Q_m}{\partial \xi _m} = \frac{1}{S_m}[\frac{1}{\xi _m^2(\lambda ) } - 1],
\quad \frac{\partial Q_m}{\partial \xi _{m+1}} = \frac{-1}{S_{m+1}}[\frac{1}{\xi _{m+1}^2(\lambda ) } - 1], \]
and all other partial derivatives are zero.
Suppose $S_m(\lambda, l_m) \not= 0$, and $\xi _m(\lambda )^2 + 1 \not = 0$ for $m=1,\dots, M$.
Then the $M-1$ gradients $\nabla Q_m$ are linearly independent, so outside of a discrete set of $\lambda $ the functions  
$\xi _1,\dots ,\xi _{K-1}$ would be holomorphic functions of $\lambda ,\xi _M$; that is, the solution set would have dimension $2$,
contradicting \corref{localman} which showed the dimension is $1$.

\end{proof}

\subsection{Extension of multipliers to $[0,\infty )$ and the spectrum}

The mapping $z \to \lambda $ given by 
\[ \lambda  = \Bigl ( \frac{1-z}{1+z} \Bigr )^2, \quad z = \frac{1-\sqrt{\lambda }}{1+\sqrt{\lambda }} \]
is a conformal map from the unit disc $\{ |z| < 1 \}$ onto 
$\lambda \in \complex \setminus [0,\infty )$.   
By using this conformal map and \lemref{multsize}
the functions $\mu _m(\lambda (z) )$ may be considered as bounded holomorphic functions on the unit disc.
Classical results in function theory \cite[p. 38]{Hoffman} insure that  $\mu _m(\lambda (z) )$ has nontangential limits almost everywhere
as a function of $z$, and so the limits
\begin{equation} \label{limeqn}
\mu _m^{\pm}(\sigma ) = \lim _{\epsilon \to 0^+} \mu _m(\sigma \pm i \epsilon ) 
\end{equation}
exist almost everywhere on $[0,\infty )$.  
By \eqref{recurse3} and \eqref{Z0}, the values $\mu _m^{\pm}(\sigma )$ are bounded away from zero uniformly on compact subsets of $\complex \setminus Z_0$. 

Since the functions $\mu _m(\lambda )$ satisfy the polynomials equations $p_m(\xi _m) = 0$, more information about 
$ \mu _m^{\pm}(\sigma )$ is available.  The equations $p_m(\xi _m) = 0$ have entire coefficients
and positive degree for $\lambda \in \complex \setminus [0,\infty )$.  Let $Z_m \subset \complex $ denote the discrete set  
where the leading coefficient vanishes.  A contour integral computation which is a variant of the Argument Principle, \cite[p. 152]{Ahlfors}
or problem 2 of \cite[p. 174]{GK},  shows that for $\lambda \in \complex \setminus Z_m$  the roots of $p_m(\lambda )$, in particular $\mu _m(\lambda )$,
are holomorphic as long as the root is simple.  For $\lambda \in \complex \setminus Z_m$ the roots extend continuously to $\lambda \in [0,\infty )$ even if the
roots are not simple.  The limiting values $\mu ^{\pm} (\sigma )$ need not agree; let 
\[\delta _m(\sigma ) = \mu ^+_m(\sigma ) - \mu ^-_m(\sigma ), \quad \sigma \in [0,\infty ) .\]

\begin{prop}
If $\mu ^+_m(\sigma ) = \mu ^-_m(\sigma )$ for $\sigma \in (\alpha ,\beta ) \setminus Z_m$,
then $\mu _m(\lambda )$ extends holomorphically across $(\alpha , \beta )$. 
\end{prop}

\begin{proof}
On any subinterval $ (\alpha _1,\beta _1) \subset (\alpha ,\beta )$ where $\mu _m(\lambda )$ extends continuously to the common value $\mu _m^{\pm}(\sigma )$,
the extension is holomorphic by Morera's Theorem \cite[p. 121]{Gamelin}.  The points in the discrete set $Z_m \cap (\alpha , \beta )$ appear to be possible obstacles
to the existence of a holomorphic extension, but since the extended function $\mu _m(\lambda )$ is bounded 
the extension can be continued holomorphically across $Z_m \cap (\alpha , \beta )$ too by Riemann's Theorem on removable singularities.
\end{proof}

\begin{thm} \label{holex}
Assume $\sigma \notin Z_0$.  For $m = 1,\dots , M$  suppose $\mu _m^{\pm } (\sigma ) \not= \pm 1$ and 
$\mu _m(\lambda )$ extends holomorphically (resp. continuously) to $\sigma \in \real $ from above (resp. below).
Then the kernel function $R_e(x,t,\lambda )$ of \eqref{rkernel} extends 
holomorphically (resp. continuously) from above (resp. below) to 
\[ R_e^{\pm} (x,t,\sigma ) = \lim_{\epsilon \to 0^+} R_e(x,t,\sigma + i \epsilon ), \quad \sigma \in \real .\]   
\end{thm}

\begin{proof}
The Wronskian formula \eqref{Weval} shows that $1/W_m(\lambda )$ extends holomorphically (resp. continuously) if $\mu _m(\lambda )$
does and $\mu _m^{\pm } (\sigma ) \not= \pm 1$.   \thmref{Yval} shows that the vertex values $y_{\pm}(w) $ extend in the same fashion as 
the multipliers $\mu _m$.  Finally, the interpolation formula \eqref{interp} provides a holomorphic extension of $y_{\pm}$ from the vertex values as long as
$\sigma \notin Z_0$, that is $S_m(l_m,\sigma )\not= 0 $ for $m = 1,\dots , M$.
\end{proof}

Recall \cite[p. 237,264]{RS1} that if $P$ denotes the family of spectral projections for a self adjoint operator, in this case $\Delta +q$, 
then for any $f \in L^2({\cal T})$
\begin{equation} \label{Stoneform}
\frac{1}{2}[P_{[a,b]} + P_{(a,b)}]f = 
\lim_{\epsilon \downarrow 0} \frac{1}{ 2 \pi i}\int_a^b [R(\sigma  + i\epsilon ) 
- R(\sigma  - i\epsilon )]f \ d\sigma .
\end{equation} 

\begin{thm} \label{acspec}
Suppose $(\alpha ,\beta ) \cap Z_0  = \emptyset $.  For $m = 1,\dots , M$ also assume that $(\alpha ,\beta ) \cap Z_m  = \emptyset $ and 
that $\mu _m^{\pm } (\sigma ) \not= \pm 1$ for all $\sigma \in (\alpha ,\beta )$.  
If $[a,b] \subset (\alpha ,\beta )$, $e$ is an edge of type $m$, and $f \in L^2(e)$, then 
\begin{equation} \label{Sform}
P_{[a,b]} f =  \frac{1}{2 \pi i}\int_a^b [R_e^+ (\sigma ) - R_e^-(\sigma )]f \ d\sigma .
\end{equation}
If $\sigma _1 \in (\alpha , \beta )$, then $\sigma _1$ is not an eigenvalue of
$\Delta + q$. 
\end{thm}

\begin{proof}
As noted above, the assumption that $(\alpha ,\beta ) \cap Z_m  = \emptyset $ means the multipliers $\mu _m$ extend continuously to $[a,b]$ 
from above and below.  Since $\mu _m^{\pm } (\sigma ) \not= \pm 1$ the function $1/W_m(\lambda )$ extends continuously to $[a,b]$.
Based on \thmref{Yval} and the interpolation formula \eqref{interp}, the kernel $R_e(x,t,\lambda )$ described in \eqref{rkernel} 
extends continuously to $[a,b]$ from above and below.
The convergence of $R_e(x,t,\sigma \pm i \epsilon )$ to $R_e(x,t,\sigma  )$ is uniform for
$t,x$ coming from a finite set of edges. 

If the support of $g \in L^2(\tree _M) $ is contained in a finite set of edges, then \eqref{Stoneform} and the uniform convergence of $R_e(x,t,\sigma \pm i \epsilon )$ to $R_e(x,t,\sigma  )$ gives
\[ \langle \frac{1}{ 2}[P_{[a,b]} + P_{(a,b)}]f,g \rangle = \frac{1}{ 2 \pi i}\int_a^b [R_e^+ (\sigma ) - R_e^-(\sigma )]\langle f,g \rangle \ d\sigma .\]
The set of $g$ with with support in a finite set of edges is dense in $L^2(\tree _M )$, so the restriction on the support of $g$ may be dropped. 
Suppose $g$ is an eigenfunction with eigenvalue $\sigma _1 \in (a,b)$ and with $\| g \| = 1$, while $f$ is the restriction of $g$ to the edge $e$.
Then the continuity of $R_e(x,t,\sigma )$ means there is a $C_e$ such that
\[ | \frac{1}{ 2 \pi i}\int_a^b [R_e^+ (\sigma ) - R_e^-(\sigma )]\langle f,g \rangle \ d\sigma | \le C_e|b-a| .\] 
This implies  $\langle P_{\sigma _1}g,g \rangle = 0$, so the eigenfunction $g$ doesn't exist.
Finally, the absence of point spectrum in $(\alpha ,\beta )$ means that $P_[a,b] = P_(a,b)$, giving the formula \eqref{Sform}.

\end{proof}

\begin{thm}
Assume $\sigma \in [0,\infty ) \setminus Z_0$ and for $m = 1,\dots , M$  suppose $\mu _m^{\pm } (\sigma ) \not= \pm 1$. 
Then  $\sigma $ is in the resolvent set of $\Delta + q$ if and only if $\delta _m(\sigma ) = 0$ in open neighborhood of $\sigma $
for $m  = 1,\dots , M$.
\end{thm}

\begin{proof} If $\sigma $ is in the resolvent set then the kernels described in \eqref{rkernel} will have a common holomorphic extension 
to $\sigma $ from above and below.  Evaluation gives 
\[R_m(l_m,0,\lambda ) = \frac{\mu _m^2(\lambda)}{1 - \mu _m^2(\lambda )}S_m(l_m,\lambda ) =  [\frac{1}{1 - \mu _m^2(\lambda )} - 1]S_m(l_m,\lambda )\]
so that $(\mu _m^+)^2(\sigma ) = (\mu _m^-)^2(\sigma )$.  A second evaluation,
\[R_m(0,0,\lambda ) = \frac{\mu _m(\lambda)}{1 - \mu _m^2(\lambda )}S_m(l_m,\lambda ), \]
shows $\mu _m^+(\sigma ) - \mu _m^-(\sigma ) = \delta _m(\sigma ) = 0$.

Now suppose $\delta _m(\sigma ) = 0$ in open neighborhood of $\sigma $.  \thmref{holex} notes that $R_m(x,t,\lambda )$ extends holomophically 
to a neighborhood of $\sigma $.  If $g \in L^2(\graph )$ with support in a finite set of edges, the function
$\langle R_m(\lambda )f,g\rangle $ also extends holomorphically as a single valued function in an interval $(\alpha ,\beta )$ containing $\sigma $.
If $[a,b] \subset (\alpha ,\beta )$, then for any $f \in L^2(e)$
\[ \langle {1 \over 2}[P_{[a,b]} + P_{(a,b)}]f, g \rangle = 0 .\]
The set of $g$ with with support in a finite set of edges is dense in $L^2(\graph )$, so ${1 \over 2}[P_{[a,b]} + P_{(a,b)}]f = 0$
for all $f \in L^2(e)$.  By linearity $[P_{[a,b]} + P_{(a,b)}]h = 0$ for any $h \in L^2(\graph )$ with support in a finite set of edges,
and since the projections are bounded we conclude that $P_{[a,b]} + P_{(a,b)} = 0$ and $(\alpha , \beta )$ is in the resolvent set \cite[p. 357]{Kato}. 

\end{proof}

\begin{cor} \label{resreal}
Assume $\sigma \in [0,\infty ) \setminus Z_0$ and for $m = 1,\dots , M$  suppose $\mu _m^{\pm } (\sigma ) \not= \pm 1$. 
Then  $\sigma $ is in the resolvent set of $\Delta + q$ if and only if $\mu ^+_m(\sigma ) $ is real valued in open neighborhood of $\sigma $
for $m  = 1,\dots , M$.
\end{cor}

\begin{proof}
If  $\mu ^+_m(\sigma ) $ is real valued, then the symmetry $\mu _m(\overline{\lambda }) =  \overline{\mu _m(\lambda )}$ established in \lemref{holomult} 
means $\mu ^-_m(\sigma ) $ has the same real value.  The same symmetry also implies that $\delta _m(\sigma ) \not= 0$ if $\mu ^+_m(\sigma ) $ is not real valued.
\end{proof}

\begin{thm} \label{poslbnd}
For $q \ge 0$ and $M \ge 2$ the spectrum of $\Delta +q$ has a strictly positive lower bound.
\end{thm}

\begin{proof}
Since 
\[\langle (\Delta + q)f,f \rangle = \langle \Delta f,f \rangle + \int_{\tree _M} q|f|^2  ,\]
it suffices to verify the result when $q = 0$.

Consider the case when the edge lengths $l_m$ are all equal to $1$.
Then the system \eqref{recurse3} reduces to
\[(2M-1)\mu ^2(\lambda )  - 2M C(1,\lambda )\mu (\lambda ) + 1 = 0.\]
The quadratic formula gives
\[\mu (\lambda ) = \frac{2MC(1,\lambda ) \pm \sqrt{4M^2C^2(1,\lambda ) - 4(2M-1)}}{2(2M-1)}.\]
Since $C(1,0) = \cos(0) = 1$, the discriminant has the positive value
$4M^2 - 4(2M-1) = 4(M-1)^2$ when $\lambda = 0$, and $\mu (\sigma )$ is real 
as long as $\cos ^2(\sqrt{\sigma }) \ge (2M-1)/M^2$.

Returning to the general case of a graph $\graph $ with unconstrained edge lengths, recall \eqref{qform} that the quadratic form for $\Delta $ is
\[\langle \Delta f,f \rangle = \int_{\tree _M} (f')^2 .\]
Let $x$ be a coordinate for intervals $[0,l_m]$ and $t$ for the interval $[0,1]$.
For $m = 1,\dots , M$ let $x = \phi _m(t) $ be a smooth change of variables.
Assume $\phi ' \ge C_1 > 0$ and $\phi '(t) = 1$ for $t$ in neighborhoods of $0$ and $1$.
If $f$ is in the domain of $\Delta $ for the graph $\graph $, then $f\circ \phi $ will be   
in the domain of $\Delta $ for a graph $\graph _1$ whose edge lengths are all $1$.

The chain rule and the change of variables formula for integrals give
\[ \int_0^{l_m} |f(x)|^2 \ dx = \int_0^1 | f(\phi _m(t))|^2 \phi _m'(t) \ dt , \]
and
\[ \int_0^{l_m} |\frac{df(x)}{dx}|^2 \ dx  = \int_0^1 | \frac{d}{dt}f(\phi _m(t))|^2 \frac{1}{\phi _m'(t)} \ dt . \]
As a consequence there is a constant $C > 0$ such that
\[\frac{\int_{\graph} |\frac{df(x)}{dx}|^2 \ dx }{\int_{\graph } |f(x)|^2 \ dx} \ge C \frac{\int_{\graph _1} |\frac{df(\phi(t))}{dt}|^2 \ dt }{\int_{\graph _1} |f(t)|^2 \ dt}.\]
The calculation for graphs with edge lengths $1$ shows that the expression on the right has a strictly positive lower bound. 
\end{proof}

\begin{cor} \label{ratcase}
Suppose $M \ge 2$ and the lengths $l_m$ are rational.  Then the resolvent set of $\Delta $ includes an unbounded subset of $[0, \infty)$.
\end{cor}

\begin{proof}
Assume $\lambda \in \complex \setminus (-\infty ,0)$ so that $\sqrt{\lambda }$ may be taken to be continuous and positive
for $\lambda > 0$.
In case $q = 0$,
\[C(l_m, \lambda ) = \cos (l_m \sqrt{\lambda }), \quad \sqrt{\lambda } S(l_m, \lambda ) = \sin(l_m \sqrt{\lambda }),\]
and these functions are periodic in $\sqrt{\lambda }$ with period $2\pi /l_m $. 
If $l_m = \tau _m/\eta _m $ with $\tau  _m, \eta _m$ positive integers, then the functions have a common period
$p = 2\pi \prod_{m=1}^M \eta _m $.   

The functions $C(l_m, \lambda )$ and $S(l_m,\lambda )$ appear as coefficients in the equations \eqref{recurse3}. 
After multiplication by $1/\sqrt{\lambda }$, the equations \eqref{recurse3} exhibit the same periodicity,
so have identical solutions for $\lambda $ and $\lambda _1$ whenever 
$\sqrt{\lambda _1} = \sqrt{\lambda } + 2n p $ for any  positive integer $n$.  

By \thmref{intmult} the solutions of \eqref{recurse3} which are multipliers are determined by the summability condition if $\lambda \in \complex \setminus [0,\infty ) $,
so $\mu _m([\sqrt{\lambda } + 2np]^2) = \mu _m (\lambda ) $ for nonreal $\lambda $.  This identity extends by continuity to $\lambda \in [0,\infty )$.
By \thmref{poslbnd} there is a $\sigma _0 > 0$ such that $[0,\sigma _0)$ is in the resolvent set of $\Delta $.
\corref{resreal} shows that except possibly at a discrete set of points, the points $\sigma \in [0,\infty )$ which are in the resolvent set are characterized
by real values of the multipliers $\mu _m(\sigma )$, so except for a discrete set of possible exceptions,  
$(4n^2p^2 , [\sqrt{\sigma _0} + 2np]^2)$ is a subset of the resolvent set for $\Delta $.
\end{proof}

\section{Sample computations}

In this section some sample spectral computations are carried out for the case $M = 2$.
The first step is to reduce  the system of equations \eqref{recurse3} to equations for individual multipliers.
Two equations of degree four with entire coefficients are obtained.
For $q = 0$ these equations are solved numerically (using Matlab) for positive values $\sigma $ of the spectral parameter.
After eliminating spurious solutions, the multiplier data is displayed in several figures.

\subsection{Elimination step}

When $M=2$ the system of equations \eqref{recurse3} may be written as 
\begin{equation} \label{special1}
\frac{\mu _1^2(\lambda ) - 1 }{S_1(l_1,\lambda )\mu _1(\lambda )}  =  2 \frac{\mu _1(\lambda ) - C_1(l_1,\lambda )}{S_1(l_1,\lambda )}
 + 2 \frac{\mu _2(\lambda ) - C_2(l_2,\lambda )}{S_2(l_2,\lambda )} ,
 \end{equation}
\[\frac{\mu _2^2(\lambda ) - 1 }{S_2(l_2,\lambda )\mu _2(\lambda )}  =  2 \frac{\mu _1(\lambda ) - C_1(l_1,\lambda )}{S_1(l_1,\lambda )}
 + 2 \frac{\mu _2(\lambda ) - C_2(l_2,\lambda )}{S_2(l_2,\lambda )}. \]
Subtracting the second equation from the first gives
\[ \frac{\mu _1^2(\lambda ) - 1 }{S_1(l_1,\lambda )\mu _1(\lambda )}  -
\frac{\mu _2^2(\lambda ) - 1 }{S_2(l_2,\lambda )\mu _2(\lambda )}  = 0,\]
Solving this quadratic equation for $2\mu _2(\lambda )$ gives
\begin{equation} \label{special2}
2\mu _2(\lambda )   - \frac{\mu _1^2(\lambda ) - 1 }{S_1(l_1,\lambda )\mu _1(\lambda )} S_2(l_2,\lambda )
\end{equation}
\[= \pm \Bigl [ \frac{(\mu _1^2(\lambda ) - 1)^2 }{S_1^2(l_1,\lambda )\mu _1^2(\lambda )} S_2^2(l_2,\lambda ) + 4 \Bigr ]^{1/2}. \]

\eqref{special1} is already first order in $\mu _2(\lambda )$, and may be rewritten as
 \[ 2 \mu _2(\lambda )    - \frac{\mu _1^2(\lambda ) - 1 }{S_1(l_1,\lambda )\mu _1(\lambda )} S_2(l_2,\lambda ) \]
\[ = - 2 S_2(l_2,\lambda ) \frac{\mu _1(\lambda ) - C_1(l_1,\lambda )}{S_1(l_1,\lambda )}  + 2 C_2(l_2,\lambda )  .\]
Replacing the left hand side using \eqref{special2} and squaring gives
\[ \frac{(\mu _1^2(\lambda ) - 1)^2 }{S_1^2(l_1,\lambda )\mu _1^2(\lambda )} S_2^2(l_2,\lambda ) + 4 
 =    4 S_2^2(l_2,\lambda ) \Bigl [\frac{\mu _1(\lambda ) - C_1(l_1,\lambda )}{S_1(l_1,\lambda )} \Bigr]^2 \]
\[  -  8 C_2(l_2,\lambda ) S_2(l_2,\lambda ) \frac{\mu _1(\lambda ) - C_1(l_1,\lambda )}{S_1(l_1,\lambda )} + 4C_2^2(l_2,\lambda ) \]
After some clean-up we get
\[ 3S_2^2(l_2,\lambda )\mu _1^4(\lambda )  - 8 \Bigl [  S_1(l_1,\lambda )C_2(l_2,\lambda ) S_2(l_2,\lambda ) + S_2^2(l_2,\lambda )C_1(l_1,\lambda ) \Bigr ] \mu _1^3 (\lambda ) \]
\[+ \Bigl [2S_2^2(l_2,\lambda ) - 4 S_1^2(l_1,\lambda ) +  4 S_2^2(l_2,\lambda ) C_1^2(l_1,\lambda )  + 4C_2^2(l_2,\lambda )S_1^2(l_1,\lambda ) \]
\[ + 8S_1(l_1,\lambda )C_2(l_2,\lambda ) S_2(l_2,\lambda ) C_1(l_1,\lambda ) \Bigr ] \mu _1^2(\lambda ) - S_2^2(l_2,\lambda ) = 0.\]

The equation satisfied by $\mu _2(\lambda )$ is obtained by interchanging the subscripts $1$ and $2$.

\begin{figure}
\begin{center}
\includegraphics[height=75mm, width=100mm]{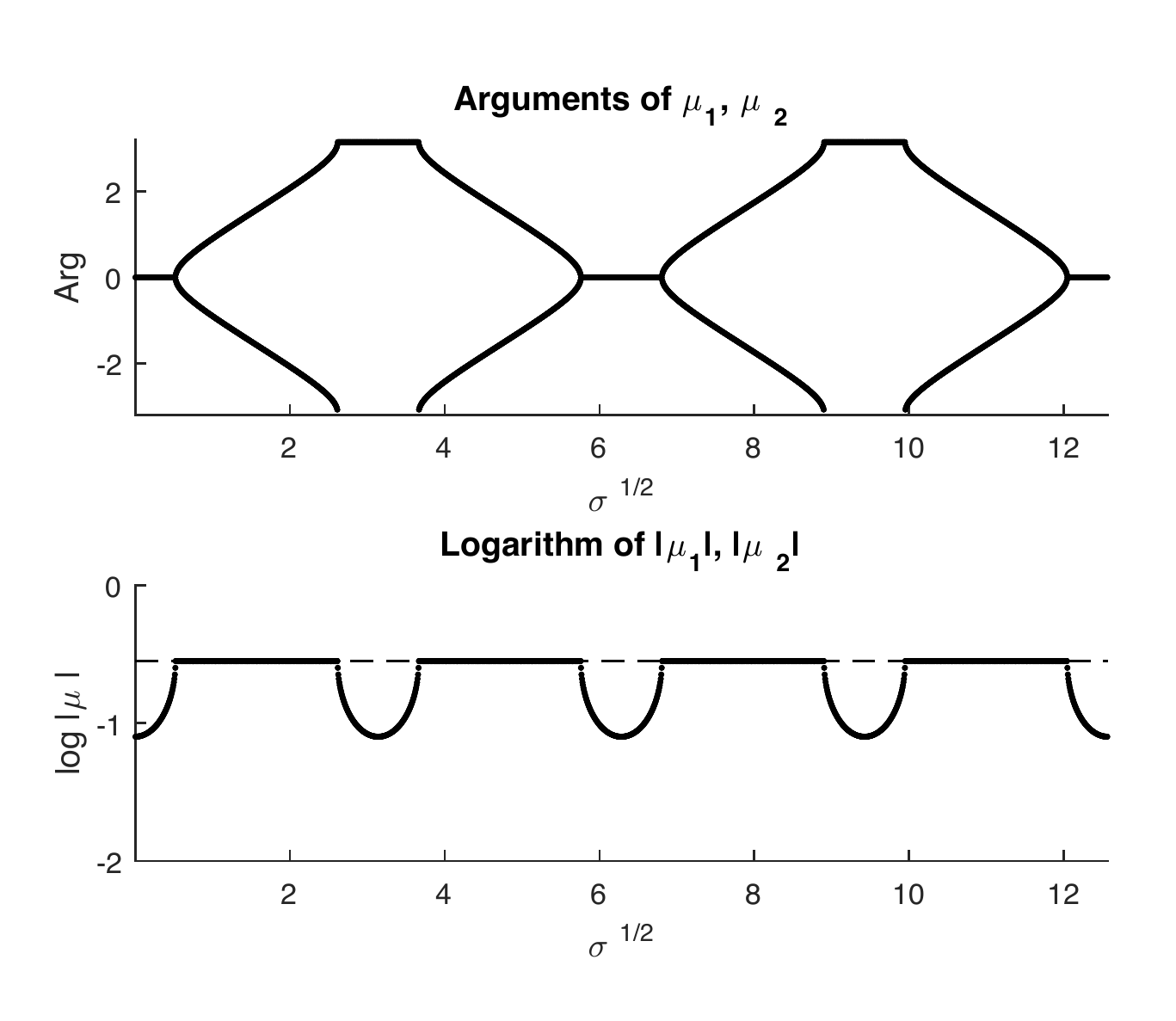}
\caption{Case $l_1 = 1$, $l_2 = 1$ } \label{arg1}
\end{center}
\end{figure}

\subsection{Numerical work}

Figures \ref{arg1}, \ref{arg89}, and \ref{arg2} display multiplier data for three cases.  In all cases 
$q = 0$ and $l_1 = 1$.  The values of $l_2$ are: (i) $l_2 = 1$, (ii) $l_2 = 0.89$, and (iii) $l_2 = 2$.

For a range of positive values of $\sigma $, solutions of the degree four polynomial equations for $\mu _1(\sigma )$ and $\mu _2(\sigma )$
are computed.  Actual multiplier pairs $(\mu _1,\mu _2)$ must satisfy the system \eqref{recurse3}, as well as the bounds
implied by the square integrability condition \eqref{summcond0}.   To eliminate spurious solutions,   
the expressions in \eqref{recurse3} were evaluated, and candidate pairs $(\mu _1, \mu _2)$ were rejected if either equation
had an expression with magnitude greater than $10^{-8}$.  Pairs were also rejected if either candidate multiplier
had $|\mu _j| > 1$, or if the minimum multiplier magnitude exceeded $1/\sqrt{3}$.

Each figure contains two parts, the multiplier arguments and the logarithm of the magnitudes.  
Figure \ref{arg1} is the case with $l_1 = l_2 = 1$.  In this case the two multipliers are equal.
By \corref{resreal}, real points in the resolvent set can be recognized by real values for both multipliers,
except when $\sigma $ lies in a discrete exceptional set.  Eigenvalues in these sets are possible, as discussed in \cite{Carlson97}.

Figure \ref{arg89} illustrates the case $l_2 = 0.89$.  When the multipliers are not real they will appear in conjugate pairs.
Unlike the classical Hill's equation, multipliers may vary in magnitude when they are not real valued. 
The multiplier arguments may exhibit occasional discontinuities.

Figure \ref{arg2} illustrates the case $l_2 = 2$.  
The multiplier argument discontinuities are clearly visible.
Notice that the horizontal axis displays $\sigma ^{1/2}$;
the predicted periodicity from the proof of \corref{ratcase} is evident.

\begin{figure}
\begin{center}
\includegraphics[height=75mm, width=100mm]{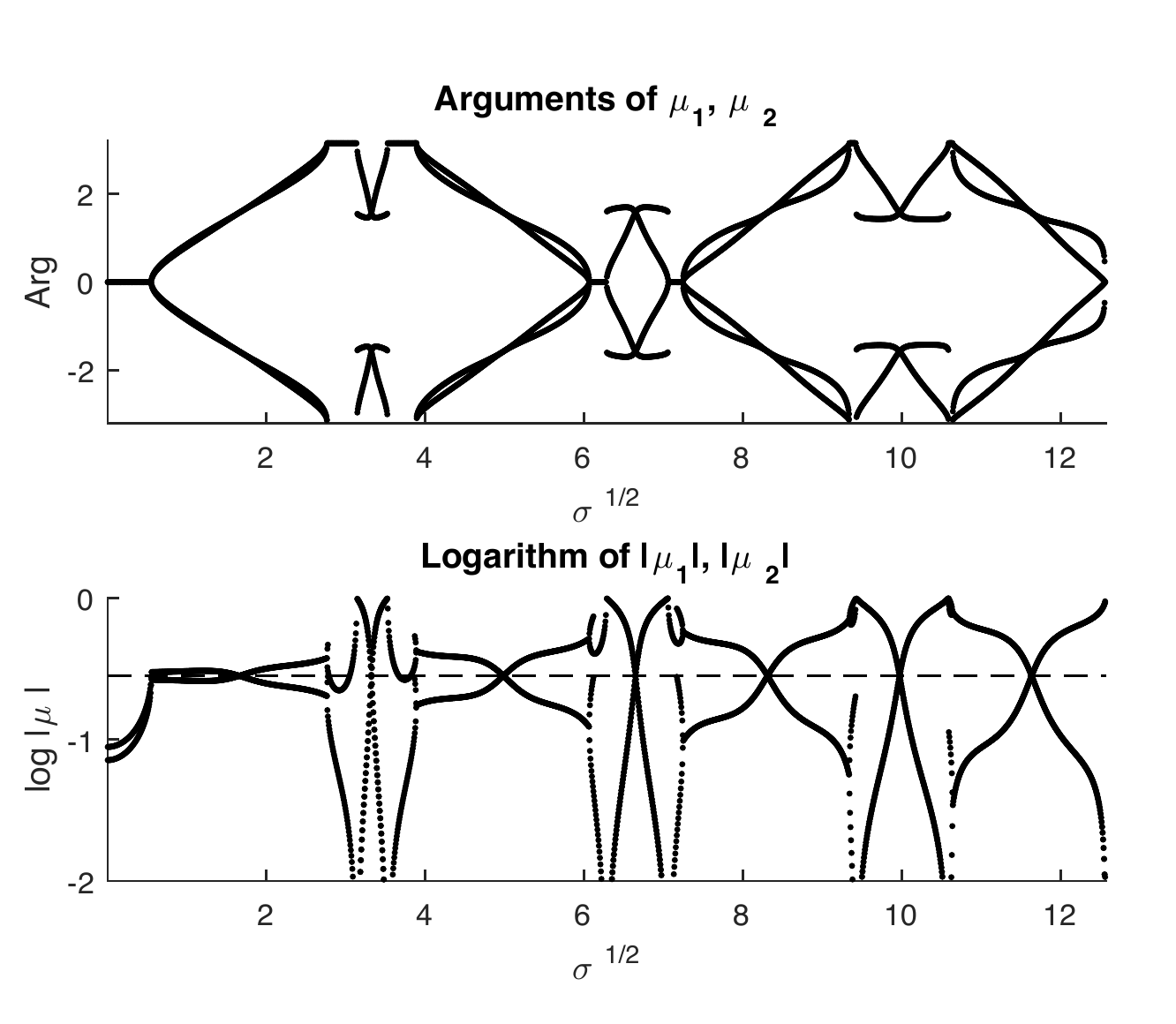}
\caption{ } \label{arg89}
\end{center}
\end{figure}

\begin{figure}
\begin{center}
\includegraphics[height=75mm, width=100mm]{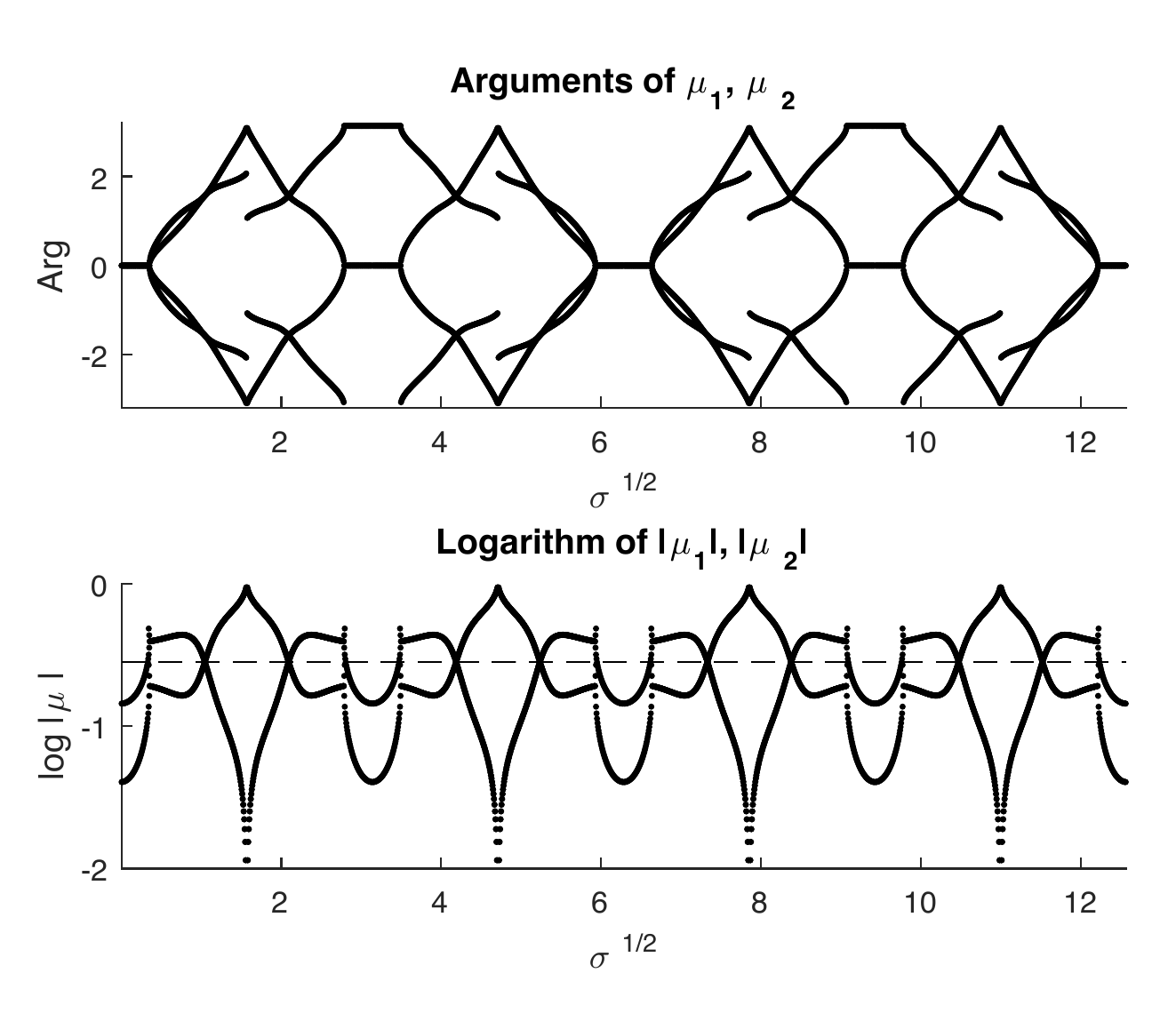}
\caption{ } \label{arg2}
\end{center}
\end{figure}

\end{document}